\documentclass{llncs}
\usepackage{stmaryrd}
\usepackage{amsmath}
\usepackage{amsfonts}
\usepackage{txfonts}
\usepackage{amssymb}
\usepackage{makeidx}  
\usepackage{amsfonts,amssymb}
\usepackage{graphicx}
\begin{document}

\title{On Unique Games with Negative Weights}
\titlerunning{On Unique Games}  
\renewcommand\thefootnote{}

\author{Peng Cui\inst{1}}

\authorrunning{Peng Cui}   
%
\tocauthor{Peng Cui}
\institute{Key Laboratory of Data Engineering and Knowledge Engineering, MOE,
School of Information Resource Management, Renmin University of China, Beijing
100872, P. R. China.\\
\email{cuipeng@ruc.edu.cn}}

\maketitle              

\begin{abstract}
In this paper, the author defines Generalized Unique Game Problem (GUGP), where
weights of the edges are allowed to be negative. Two special types of GUGP are
illuminated, GUGP-NWA, where the weights of all edges are negative, and
GUGP-PWT($\rho$), where the total weight of all edges are positive and the
negative-positive ratio is at most $\rho$. The author investigates the
counterpart of the Unique Game Conjecture on GUGP-PWT($\rho$). The author shows
that Unique Game Conjecture on GUGP-PWT(1) holds true, and Unique Game
Conjecture on GUGP-PWT($1/2$) holds true, if the 2-to-1 Conjecture holds true.
The author poses an open problem whether Unique Game Conjecture holds true on
GUGP-PWT($\rho$) with $0<\rho<1$.
\end{abstract}
\section{Introduction}
The Unique Game Conjecture (UGC) is put forward by Khot on STOC 2002 as a
powerful tool to prove lower bound of inapproximabilty for combinatorial
optimization problems\cite{k}. It has been shown by researchers a positive
resolution of this conjecture would imply improved even best possible hardness
results for many famous problems, to name a few, Max Cut, Vertex Cover,
Multicut, Min 2CNF Deletion, making an important challenge to prove or refute
the conjecture.

Some variations of UGC have been mentioned. Rao proves a strong parallel
repetition theorem which shows Weak Unique Game Conjecture is equivalent to
UGC\cite{r}. Khot et al. show that Unique Game Conjecture on Max 2LIN(q) is
equivalent to UGC\cite{kkmo}. Khot poses the d-to-1 Conjectures in his original
paper for $d\ge2$\cite{k}. O'Donell et al. show a tight hardness for
approximating satisfiable constraint satisfaction problem on 3 Boolean
variables assuming the d-to-1 Conjecture for any fixed $d$\cite{ow1}. Dinur et
al. use the 2-to-2 Conjecture to derive the
hardness results of Approximate Coloring Problem, and prove that the
2-to-1 Conjecture implies their 2-to-2 Conjecture\cite{dmr}. Guruswami et al.
use the 2-to-1 Conjecture to derive the hardness result of Maximum
k-Colorable Subgraph Problem\cite{gs}. It is unknown whether UGC implies any of
the d-to-1 Conjectures, or vice versa.

Recently, the authors of \cite{abs} designed a subexponential time algorithm
for Unique Games Problem (UGP), which challenges the position of
UGC as a tool to prove lower bound of inapproximabilty. While their
results stop short of refuting UGC, they do suggest that UGP is significantly
easier than NP-hard problems. On the other side, the authors of \cite{ow2}
determined a new point of $(c,s)$-approximation NP-hardness of UGP, compared to
\cite{h}, and their new result, together with the result of \cite{fr},
determines the two-dimensional region of all known $(c,s)$-approximation
NP-hardness of UGP.

In this paper, the author defines Generalized Unique Game Problem (GUGP), where
weights of the edges are allowed to be negative. Two special types of GUGP are
illuminated, GUGP-NWA, where the weights of all edges are negative, and
GUGP-PWT($\rho$), where the total weight of all edges are positive and the
negative-positive ratio is at most $\rho$. GUGP-PWT($\rho$) over $1\ge\rho\ge
0$ makes a possible phase transition from a 2-Prover 1-Round Game Problem with $(1-\zeta,\delta)$-approximation NP-hardness to UGP. The author shows
that UGC on GUGP-PWT(1) holds true, UGC on GUGP-PWT($1/2$) holds true, if the
2-to-1 Conjecture holds true, and the $(1-\zeta,\delta)$-approximation NP-hardness
of GUGP-PWT($\rho$) possesses the compactness property when $\rho\rightarrow
0$.

Section 2 demonstrates some definitions. The author shows the main results for
GUGP-NWA and GUGP-PWT($\rho$) in Section 3. Section 4 is some discussions.

\section{Preliminaries}
In {\it 2-Prover 1-Round Game Problem (2P1R)}, we are given a bipartite graph
$G=(V,W;E)$, with each edge $e$ having a weight $w_{e}\in \mathbb{Q}^{+}$. We
are also given two sets of labels, $k_{1}$ and $k_{2}$, which we identify with
$[k_{1}]=\{1,\cdots,k_{1}\}$ and $[k_{2}]=\{1,\cdots,k_{2}\}$. Each edge
$e=(u,v)$ in the graph is equipped with a relation
$R_{e}\subseteq[k_{1}]\times[k_{2}]$. The solution of the problem is a labeling
$f_{1}:V\rightarrow [k_{1}]$ and $f_{2}:W\rightarrow [k_{2}]$ which assigns a
label to each vertex of $G$. An edge $e=(u,v)$ is said to be satisfied under
$f_{1}$ and $f_{2}$ if $(f_{1}(u),f_{2}(v))\in R_{e}$, else is said to be
unsatisfied. The object of the problem is to find a labeling maximizing the
total weight of the satisfied edges. The value of the instance, $Val(G)$, is
defined as the maximum total weight of the satisfied edges divided by the total
weight of all edges.

{\it Unique Game Problem (UGP)} can be viewed as a special type of 2P1R. In
UGP, we are given a graph $G=(V,E)$, a weight function $w_{e}\in
\mathbb{Q}^{+}$ for $e\in E$, and a set of labels, $[k]$. Each edge $e=(u,v)$
in the graph is equipped with a permutation $\pi_{e}:[k]\rightarrow[k]$. The
solution of the problem is a labeling $f:V\rightarrow [k]$ which assigns a
label to each vertex of $G$. An edge $e=(u,v)$ is said to be satisfied under
$f$ if $\pi_{e}(f(u))=f(v)$, else is said to be unsatisfied. Note that we allow
$G$ is a graph with parallel edges, i.e., there exist more than one edge
between two vertices.

It is possible to define two optimization problems in this situation. In {\it
Max UGP}, the value of the instance is defined as the maximum total weight of
the satisfied edges divided by the total weight of all edges. In {\it Min UGP},
the value of the instance is defined as the minimum total weight of the
unsatisfied edges divided by the total weight of all edges.

Khot initiates much of the interest in the following conjecture by showing that
many hardness results stem from it. It basically states that it is NP-hard to
distinguish whether many or only few edges are satisfied.\\

\noindent{\bf Conjecture 1. (\cite{k}\ Unique Game Conjecture in Max UGP Form)}
{\it For every $\zeta,\delta>0$, there is a $k=k(\zeta,\delta)$ such that given
an instance $G$ of Max UGP with $k$ labels it is NP-hard to
distinguish whether $Val(G)>1-\zeta$ or $Val(G)<\delta$.}\\

The conjecture can be restated in Min UGP form, and the two conjectures are
equivalent.\\

\noindent{\bf Conjecture 2. (Unique Game Conjecture in Min UGP Form)} {\it For
every $\zeta,\delta>0$, there is a $k=k(\zeta,\delta)$ such that given an
instance $G$ of Min UGP with $k$ labels it is NP-hard to distinguish whether
$Val(G)<\zeta$ or $Val(G)>1-\delta$.}\\

The {\it $(c,s)$-approximation NP-hardness} of Max UGP is defined as: for some
fixed $0<s<c<1$, there is a $k$ such that given an instance $G$ of Max UGP with $k$
labels it is NP-hard to distinguish whether $Val(G)\ge c$ or
$Val(G)<s+\varepsilon$ for any $\varepsilon>0$.

2-to-1 Game and 2-to-2 Game are two special types of 2P1R. In {\it 2-to-1
Game}, we are given a bipartite graph $G=(V,W;E)$, with each edge $e$ having a
weight $w_{e}\in \mathbb{Q}^{+}$. We are also given two sets of labels, $[2k]$
for $V$ and $[k]$ for $W$. Each edge $e=(u,v)$ in the graph is equipped with a
2-to-1 projection. A projection $\sigma:[2k]\rightarrow [k]$ is said to be a
2-to-1 projection if for each element $j\in [k]$ we have $|\sigma^{-1}(j)|=2$.
The value of the instance of 2-to-1 Game, $Val(G)$, is defined as the maximum
total weight of the satisfied edges divided by the total weight of all edges.

In {\it 2-to-2 Game}, we are given a graph $G=(V,E)$, a weight function
$w_{e}\in \mathbb{Q}^{+}$ for $e\in E$, and a set of labels, $[k]$. Each edge
$e=(u,v)$ in the graph is equipped with a 2-to-2 relation. A relation
$R\subseteq[2k]\times[2k]$ is said to be a 2-to-2 relation if there are two
permutations $\pi_{u},\pi_{v}:[2k]\rightarrow [2k]$ such that $(i,j)\in R$ iff
$(\pi_{u}(i),\pi_{v}(j))\in T$ where
$$T:=\bigcup_{l=1}^{k}{\{(2l-1,2l-1),(2l-1,2l),(2l,2l-1),(2l,2l)\}}.$$
The value of the instance of 2-to-2 Game, $Val(G)$, is defined as the maximum
total weight of the satisfied edges divided by the total weight of all edges.

The author lists the 2-to-1 Conjecture and the 2-to-2 Conjecture in
their maximization forms. Note that the latter is somewhat different from that in
\cite{dmr}. It can be proved that the 2-to-1 Conjecture implies the 2-to-2
Conjecture along the line of \cite{dmr}.\\

\noindent{\bf Conjecture 3. (2-to-1 Conjecture)} {\it For every $\delta>0$,
there is a $k=k(\delta)$ such that given an instance $G$ of 2-to-1 Game with
the label sets $[2k]$ and $[k]$ it is NP-hard to distinguish whether $Val(G)=1$
or $Val(G)<\delta$.}\\

\noindent{\bf Conjecture 4. (2-to-2 Conjecture)} {\it For every $\delta>0$,
there is a $k=k(\delta)$ such that given an instance $G$ of 2-to-2 Game with the
label set $[2k]$ it is NP-hard to distinguish whether $Val(G)=1$ or
$Val(G)<\delta$.}\\

In this paper, the author defines {\it Generalized Unique Game Problem (GUGP)},
where the weights of edges are allowed to be negative. In GUGP, we are given a
graph $G=(V,E)$, possibly having parallel edges, a weight function $w_{e}\in
\mathbb{Q}$ for $e\in E$, and a set of labels, $[k]$. Each edge $e=(u,v)$ in
the graph is equipped with a permutation $\pi_{e}:[k]\rightarrow[k]$. The
solution of GUGP is a labeling $f:V\rightarrow [k]$ which assigns a label to
each vertex of $G$. The goal of the problem is to maximize the total weight of
the satisfied edges. Note that $w_{e}$ could be positive or negative. The author assumes
there is no edge with zero weight for sake of clearance.

Let $W^{+}_{G}$ be the total of the positive weights of all edges, $W^{-}_{G}$
be the total of the negative weights of all edges, and
$\Sigma_{G}=W^{+}_{G}+W^{-}_{G}$ be the total weight of all edges. The author calls
$r_{G}=|W^{-}_{G}|/W^{+}_{G}$ the negative-positive ratio of the instance.

GUGP-NWA and GUGP-PWT are two special types of GUGP. In {\it GUGP-NWA}, the
weight of all edges are negative. In {\it GUGP-PWT}, the total weight of all
edges is positive. It is possible to define two optimization problems for
GUGP-NWA and for GUGP-PWT.

In {\it Max GUGP-NWA}, we seek to minimize the total weight of the unsatisfied
edges, i.e. to maximize the absolute value of the total weight of the
unsatisfied edges. The value of Max GUGP-NWA is defined as the maximum absolute
value of the total weight of the unsatisfied edges divided by $|W^{-}_{G}|$. In
{\it Min GUGP-NWA}, we seek to maximize the total weight of the satisfied
edges, i.e. to minimize the absolute value of the total weight of the satisfied
edges. The value of Min GUGP-NWA is defined as the minimum absolute value of
the total weight of the satisfied edges divided by $|W^{-}_{G}|$.

In {\it Max GUGP-PWT}, we seek to maximize the total weight of the satisfied
edges. The value of Max GUGP-PWT is defined as the maximum total weight of the
satisfied edges divided by $\Sigma_{G}$. In {\it Min GUGP-PWT}, we seek to
minimize the total weight of the unsatisfied edges. The value of Min GUGP-PWT
is defined as the minimum total weight of the unsatisfied edges divided by
$\Sigma_{G}$. In an instance $G$ of Min GUGP-PWT, let $W_{G}(f)$ be the total
weight of the unsatisfied edges under labeling $f$, and let the optimal
labeling be $f^{*}$. The value of the instance is
$Val(G)=W_{G}(f{^{*}})/\Sigma_{G}$.

The author reminds the reader that the value of Max GUGP-PWT could be more than 1 and
and the value of Min GUGP-PWT could be less than 0.

The author defines {\it Max/Min GUGP-PWT($\rho$)} as the subproblem of Max/Min GUGP-PWT
where the negative-positive ratio of the instances is upper bounded by $\rho$,
where $\rho$ is a constant independent from $k$. Since the negative-positive
ratio is always less than 1, we set the range of $\rho$ to be $0\le\rho\le 1$.
Note that Max/Min GUGP-PWT(0) is simply Max/Min UGP.

The author gives the two equivalent counterparts of the Unique Game Conjecture on Max
GUGP-PWT($\rho$) and Min GUGP-PWT($\rho$)) as follows:\\

\noindent{\bf Conjecture 5. (Unique Game Conjecture on Max GUGP-PWT($\rho$))}
{\it For every $\zeta,\delta>0$, there is a $k=k(\zeta,\delta)$ such that given
an instance of Max GUGP-PWT($\rho$) with $k$ labels it is NP-hard to
distinguish whether $Val(G)>1-\zeta$ or $Val(G)<\delta$.}\\

\noindent{\bf Conjecture 6. (Unique Game Conjecture on Min GUGP-PWT($\rho$))}
{\it For every $\zeta,\delta>0$, there is a $k=k(\zeta,\delta)$ such that given
an instance of Min GUGP-PWT($\rho$) with $k$ labels it is NP-hard to
distinguish whether $Val(G)<\zeta$ or $Val(G)>1-\delta$.}\\

The conjectures states it is NP-hard to distinguish the following two cases:
there is a labeling under which the absolute value of the total of the negative
weight of the unsatisfied edges is almost no less than the total of the
positive weight of the unsatisfied edges; under any labeling the absolute value
of the total of the negative weight of the satisfied edges is almost no less
than the total of the positive weight of the satisfied edges.

\section{Main Results}
\subsection{GUGP-NWA}
Max GUGP-NWA can be restated as the following 2P1R. We are given a graph
$G=(V,E)$, a weight function $w_{e}\in \mathbb{Q^{+}}$, and a set of labels,
$[k]$. Each edge $e=(u,v)$ in the graph is equipped with a relation
$\bar\pi_{e}=[k]\times[k]-\pi_{e} $, where $\pi_{e}:[k]\rightarrow[k]$ is a
permutation. The solution of the problem is a labeling $f:V\rightarrow [k]$
which assigns a label to each vertex of $G$. An edge $e=(u,v)$ is said to be
satisfied under $f$ if $(f(u),f(v))\in \bar\pi_{e}$. The value of the instance
is defined as the total weight of the satisfied edges divided by the total
weight of all edges.

Since a random labeling satisfies an expectation of $1-1/k$ fraction of the
total weight of all edges, GUGP-NWA cannot have a large gap. We can prove that
it is NP-hard to approximate Min GUGP-NWA within any $poly(k)$, by the similar
arguments in the following theorem. The author omits the full proof for clarity
of the paper.

\begin{theorem}
It is NP-hard to approximate Min GUGP-NWA with the label set $[n]$ within any
$poly(n)$.
\end{theorem}
\begin{proof}
Min GUGP-NWA can be restated as: In the situation of UGP, the goal is to find
minimum fraction of the total weight of the satisfied edges. The author constructs
an approximation ratio preservation reduction from TSP to the above problem.

Given an instance of TSP problem $G=(V,E)$, where each edge of $E$ has a weight
$w_{e}\in \mathbb{Q}^{+}$. Denote $n:=|V|$. The instance of the restated form
of Min GUGP-NWA is a graph $G'=G'(V,E')$, with each edge $e'\in E'$ having a
weight $w'(e')$, and with the labeling set $[n]$. For each edge $e=(u,v)\in E$,
there are three parallel edges $e^{=}$, $e^{+}$ and $e^{-}$ between $u$ and $v$
in $E'$. $e^{=}$ has weight $M$ and equipped with permutation
$\pi^{=}=\{(1,1),(2,2),\cdots,(n,n)\}$. Let $M=n\cdot Max(w)$, where $Max(w)$
is the maximum weight of all edges in $G$. $e^{+}$ has weight $w_(e)$ and
equipped with permutation $\pi^{+}=\{(1,2),(2,3),\cdots,(n,1)\}$. $e^{-}$ has
weight $w_(e)$ and equipped with permutation
$\pi^{-}=\{(1,n),(2,1),\cdots,(n,n-1)\}$.

Given a solution of TSP problem, a Hamiltonian cycle $C$, we can assign label
$1$ to $n$ to vertices of $C$ along $C$ in $G'$, and the total weight of
satisfied edges in $G'$ is exactly the total weight of edges on $C$ in $G$.

In the other direction, given a labeling $f$ of $G'$, if there are two vertices
assigned with the same label, the total weight of the satisfied edges is at
least $M$. Otherwise all vertices are assigned with label from $1$ to $n$
respectively, let $u_{i}$ be the vertices assigned label $i$ for $1\le i\le n$,
and $e^{+}_{i}\in E'$ be the edge between $u_{i}$ and $u_{i\mod n+1}$
equipped with permutation $\pi^{+}$. The total weight of the satisfied edges is
equal to $\sum_{1\le i\le n}{w'(e^{+}_{i})}$. Let $C$ be the Hamiltonian cycle
of $G$ which consists of vertices from $u_{1}$ to $u_{n}$, then the total
weight of $C$ in $G$ is exactly the total weight of satisfied edges under $f$
in $G'$.
\end{proof}

\subsection{GUGP-PWT($\rho$)}
By \cite{fr}, for any constant $C>0$, there is a $\varepsilon>0$, such that
$(C\varepsilon,\varepsilon)$-approximating UGP is NP-hard. Note that
$C\varepsilon\rightarrow 0$, when $\varepsilon\rightarrow 0$. We get an
instance of GUGP-PWT($1-C\varepsilon$), by supplementing the graph in the
instance of UGP with two new vertices and $k$ parallel edges between the two
vertices, each of which has the negative weight $-\frac{1-C\varepsilon}{k}$ and
is equipped with one of the permutations $\pi_{i}:j\mapsto (j+i-2)\mod k+1$ for
$1\le i\le k$.

For every $\zeta,\delta>0$, we can determine $C$ such that
$\frac{1}{C}<\delta$, and the size of the label set of the instance, $k$,
satisfies $\frac{1}{kC\varepsilon}<\zeta$. Then for this instance of Max
GUGP-PWT($1-C\varepsilon$) with $k$ labels, $G$, it is NP-hard to distinguish
whether $Val(G)>1-\zeta$ or $Val(G)<\delta$. Therefore, Conjecture 6 holds true
for $\rho=1$.

Conjecture 4 implies the $(1/2,0)$-approximation NP-hardness of Max UGP by the following reduction.
Write each edge in an instance of the 2-to-2 game into two parallel edges, each
of which has the same weight as the original edge and is equipped with one of
the two disjoint permutations extracted from the original 2-to-2 relation. Thus
Conjecture 4 implies Conjecture 6 holds true for $\rho=1/2$ by the similar
technique in the first paragraph of this subsection. Since Conjecture 3 implies
Conjecture 4, Conjecture 6 holds true for $\rho=1/2$ if Conjecture 3 holds
true.

Finally, we establish a connection from Conjecture 6 to the Unique Game
Conjecture by the following theorem.

\begin{theorem}
If Conjecture 6 holds true for any $\rho>0$, Conjecture 2 holds true.
\end{theorem}
\begin{proof}
Suppose Conjecture 2 holds false, then for some $\zeta,\delta>0$, for any label
size $k$, we can decide in polynomial time whether an instance of Min UGP with
$k$ labels has a value more than $1-\delta$ or less than $\zeta$. The author claims
Conjecture 6 for $\rho=\min(\zeta,\delta)/2$ holds false.

Given an instance $G=(V,E)$ of Min GUGP-PWT($\rho$), we construct an instance
$G'=(V,E')$ of Min UGP as follows. Let $E'$ be the set of the edges in $E$ with
positive weights. Let $f^{*}$ be the optimal labeling of $G$, and $f'$ be the
optimal labeling of $G'$. Then $Val(G')=W_{G'}(f')/W^{+}_{G}$ and
$Val(G)=W_{G}(f^{*})/\Sigma_{G}$.

Since $W_{G'}(f')\ge W_{G}(f')\ge W_{G}(f^{*})$ and $\Sigma_{G}/W^{+}_{G}\ge
1-\rho$, $Val(G')\ge (1-\rho)Val(G)$.

By the definition of $E'$, $W_{G}(f^{*})\ge W_{G'}(f^{*})-\rho W^{+}_{G}$. We
have $Val(G)\Sigma_{G}=W_{G}(f^{*})\ge W_{G'}(f^{*})-\rho W^{+}_{G}\ge
W_{G'}(f')-\rho W^{+}_{G}=(Val(G')-\rho)W^{+}_{G}$. Therefore, $Val(G')\le
Val(G)+\rho$.

If $Val(G)<\zeta/2$, then $Val(G')<\zeta$. If $Val(G)>1-\delta/2$, then
$Val(G')>1-\delta$. Thus we can decide in polynomial time whether the instance
of Min GUGP-PWT($\rho$) has a value more than $1-\delta/2$ or less than
$\zeta/2$.
\end{proof}

\section{Discussions}
It leaves as an open problem whether Unique Game Conjecture holds true on
GUGP-PWT($\rho$) for $0<\rho<1$. The author makes a reasonable and rather bold
conjecture: If Conjecture 6 holds true for some $0<\rho<1$, it holds true for
any $0<\rho<1$, which would lead to the corollary that the 2-to-1 Conjecture
implies the Unique Game Conjecture. To confirm our
conjecture, it would be critical to seek techniques to derive the
$(1-\zeta,\delta)$-approximation NP-hardness on smaller $\rho$ by the
$(1-\zeta,\delta)$-approximation NP-hardness on larger $\rho$.

\end{document}